\documentclass[draftcls,12pt,onecolumn]{IEEEtran}

\usepackage{times,amsmath,color,amssymb,graphicx,epsfig,cite,psfrag,enumerate}
\usepackage{subfigure,multirow,bm}
\newtheorem{Remark}{\it Remark}[section]
\newtheorem{Theorem}{\it Theorem}[section]

\newtheorem{Lemma}{\it Lemma}[section]

\begin{document}
\title{Asymptotic Capacity of Large Relay Networks with Conferencing Links}


%\author{\IEEEauthorblockN{Chuan Huang,~\IEEEmembership{Student Member,~IEEE,}~Jinhua Jiang,~\IEEEmembership{Member,~IEEE,} \\ Shuguang
%Cui},~\IEEEmembership{Member,~IEEE,}
%
%\thanks{Chuan Huang and Shuguang Cui are with the Department of Electrical
%and Computer Engineering, Texas A\&M University, College Station,
%TX, 77843. Emails: \{huangch, cui\}@tamu.edu.}
%
%\thanks{Jinhua Jiang is with the Department of Electrical Engineering,
%Stanford University, Stanford, CA, 94305, Email:
%jhjiang@stanford.edu.}}

\author{\IEEEauthorblockN{Chuan Huang,~\IEEEmembership{Student Member,~IEEE,}~Jinhua Jiang,~\IEEEmembership{Member,~IEEE,} \\ Shuguang
Cui},~\IEEEmembership{Member,~IEEE}

\thanks{Chuan Huang and Shuguang Cui are with the Department of Electrical
and Computer Engineering, Texas A\&M University, College Station,
TX, 77843. Emails: \{huangch, cui\}@tamu.edu.}

\thanks{Jinhua Jiang is with the Department of Electrical Engineering,
Stanford University, Stanford, CA, 94305, Email:
jhjiang@stanford.edu.}}

%
%\author{\IEEEauthorblockN{Chuan Huang}
%\IEEEauthorblockA{Department of ECE\\
%Texas A\&M University\\
%College Station, TX 77843, USA\\
%Email: huangch@tamu.edu } \and \IEEEauthorblockN{Jinhua Jiang}
%\IEEEauthorblockA{Department of EE\\
%Stanford University\\
%Stanford, CA 94305, USA\\
%Email: jhjiang@stanford.edu } \and \IEEEauthorblockN{Shuguang Cui}
%\IEEEauthorblockA{Department of ECE\\
%Texas A\&M
% University\\
%College Station, TX 77843, USA\\
%Email: cui@ece.tamu.edu } }

\maketitle
\begin{abstract}
%\boldmath
In this correspondence, we consider a half-duplex large relay network, which consists of one source-destination pair and $N$ relay nodes, each of which is connected with a subset of the other relays via signal-to-noise ratio (SNR)-limited out-of-band conferencing links. The asymptotic achievable rates of two basic relaying schemes with the ``$p$-portion'' conferencing strategy are studied: For the decode-and-forward (DF) scheme, we prove that the DF rate scales as $\mathcal{O} \left( \log (N) \right)$; for the amplify-and-forward (AF) scheme, we prove that it asymptotically achieves the capacity upper bound in some interesting scenarios as $N$ goes to infinity.
\end{abstract}



\section{Introduction}

The concept of relay has already been adopted in most beyond-3G wireless technologies such as WiMAX and 3GPP UMTS Long Term Evolution (LTE) to provide coverage extension and increase capacity. From the information-theoretical viewpoint, the capacity bounds of the traditional three-node relay channel have been well studied under both the full-duplex mode \cite{cover} and the half-duplex mode \cite{madsen}, and various achievable schemes, such as decode-and-forward (DF), compress-and-forward (CF), and amplify-and-forward (AF) have been proposed. For the four-node diamond relay channel, the capacity upper bounds and achievable rates were first investigated in \cite{schein}, and then in \cite{xue,resaei}.

For the large relay networks with $N$ relay nodes, the asymptotic capacity bounds were studied in \cite{Gastpar1,blocskei,chuan1,coso}.
Considering the joint source channel coding problem for a special class of
Gaussian relay networks \cite{Gastpar1}, the
capacity upper bound is asymptotically achieved by the AF relaying scheme
as the number of relays tends to infinity. For general Gaussian relay networks,
the authors in \cite{blocskei} obtained the achievable rate scaling law for the multiple-input and multiple-output (MIMO) relay networks with AF: For the coherent relaying case, with full forward-link channel state information (CSI) in the relays, the AF achievable rate scales as $\mathcal{O} \left( \log (N) \right)$; for the noncoherent relaying case with zero forward-link CSI in the relays, it scales as $\mathcal{O} \left( \log (1) \right)$. In \cite{coso}, the authors studied the
scaling laws of the DF, CF, and linear relaying schemes, and proved that the DF rate scales at most as $\mathcal{O} \left( \log \left( \log(N) \right) \right)$ for the coherent relaying scheme. The authors in \cite{chuan1} mainly focused on the noncoherent case, and proved that the DF relaying scheme asymptotically achieves the capacity upper bound.

In practical wireless communication systems, some nodes might have the capability to exchange certain information with other nodes via extra out-of-band connections, e.g., through internet, WiFi, optical fiber, etc. From the information-theoretical viewpoint, such kind of interaction can be modeled as node conferencing \cite{cmac,ron,maric,Denus}. Specifically, for the multiple access channel (MAC) \cite{cmac}, encoder conferencing was used to exchange part of the source messages. For the broadcast channel (BC) in \cite{ron}, the decoders were designed to first compress the received signals, and then transmit the corresponding binning index numbers to the other through the conferencing links. Moreover, in \cite{maric} and \cite{Denus}, the achievable rates of compound MAC with transmitter and receiver conferencing were discussed, respectively, and some capacity results for the degraded cases were established.

In \cite{chuan2}, the authors investigated the achievable rates for the four-node diamond relay channel with rate-limited out-of-band conferencing links between the two relays, and it was shown that the DF scheme could achieve the cut-set bound even with finite conferencing link rates for the discreet memoryless channel case. In this correspondence, we extend these results to the large Gaussian relay networks with signal-to-noise ratio (SNR)-limited conferencing links among the relays, and focus on the asymptotic achievable rates of the DF and AF schemes. It is shown that the relay conferencing can improve these achievable rates, and some asymptotic capacity results can be established under certain conditions.

The rest of the paper is organized as follows. In Section II,
we introduce the assumptions and channel models. In Section III, we discuss the DF and AF achievable rates. In Section IV, we present some simulation and numerical results. Finally, the paper is concluded in Section V.

Notations: $X_N \xrightarrow{w.p.1} a$ means $X_N \rightarrow a$ with probability 1, as $N \rightarrow +\infty $; $A_N \sim B_N$ means $\lim_{N \rightarrow +\infty} |A_N - B_N| = 0$; $y_N \sim \mathcal{O} \left( \log (x_N) \right)$ means $\lim_{N \rightarrow + \infty} \frac{x_N}{y_N} = c$, where $c$ is a positive constant.

\section{Assumptions and System Model}

In this paper, we consider a large relay network with out-of-band conferencing links among the relays, as shown in Fig. \ref{fig1}, which contains one source-destination pair, and $N$ relays. We assume that there is no direct link between the source and destination. The relay nodes work in a half-duplex mode: The source transmits and the relays listen in the first time slot; the relays simultaneously transmit and the destination listens in the second time slot. For simplicity, we allocate equal time durations to the two hops \cite{blocskei,chuan1}.

The time scheduling of the transmissions at the source, relays, and conferencing links is shown in Fig. \ref{fig2}. Note that the conferencing links use out-of-band connections in relative to the source-relay links; thus, they can be allocated the same time slot. Due to the relay conferencing, there will be a one-block delay between the transmissions at the source and the relays, which requires the relays to buffer one block of source signals for each relaying operation. Assume that during each data block, the communication rate is $R$, and we need to transmit $B$ blocks in total. Thus, the average information rate is $R\frac{B}{B+1}\rightarrow R$, as $B$ goes to infinity, such that the effect of the one-block delay is negligible. In this correspondence, we focus on the one-block transmission to study the associated relaying and conferencing schemes without specifying the delay in the proof of the achievability.

We assume that each relay can conference with a subset of other relays via wired links. In this correspondence, we adopt a deterministic ``$p$-portion conferencing'' scheme: each relay can conference with other $M$ relays, and
\begin{align} \label{p_conf}
\lim_{N \rightarrow +\infty} \frac{M+1}{N} = p.
\end{align}
Without loss of generality, we assume that the $i$-th relay forwards its received signal to the relays with indices $(i+k),i=0,1,\cdots,N-1,$ and $k=1,\cdots,M$, via the conferencing links. With a little abuse of notation, we use $(i+k)$ to denote the $(i+k)_N$-th relay, where $(\cdot)_N$ means the modula over $N$. Particularly, when $N=M+1$, we call the scheme as ``complete conferencing''. Note that there exist many other conferencing schemes, i.e., random conferencing with any other $M$ relays, while the $p$-portion deterministic conferencing scheme is adopted here to simplify the analysis and provide a tractable achievable rate.

We further define the following channel input-output relationship. In the first hop, the received signal $y_i$ at the $i$-th relay, $i=1,2,\cdots,N$, is given as
\begin{eqnarray} \label{t1}
y_i=\sqrt{P_{S}}h_{i}x+n_i,
\end{eqnarray}
where $x$ is the signal transmitted by the source, $P_s$ is the transmit power at the source node, $h_{i}$ is the complex channel gain
of the $i$-th source-to-relay link, which is assumed known to the source, and $n_i$'s are the independently and identically distributed (i.i.d.) circularly symmetric complex Gaussian (CSCG) noise with distribution $\mathcal{CN}(0, N_0)$. Note that there are no particular assumptions on the distributions of $h_i$'s, which are just assumed to be independent, of zero-mean, and with uniformly and positively bounded second-order and fourth-order statistics, i.e., $ 0<b_1 \leq \mathbb{E} \left( |h_i|^2 \right) \leq b_2 <+\infty$ and $ 0<c_1 \leq \mathbb{E} \left( |h_i|^4 \right) \leq c_2 <+\infty$.

For the conferencing links, the received signal from the $i$-th relay to the $(i+k)$-th relay is given as
\begin{align} \label{tc}
y_{i,i+k}= \sqrt{ \frac{ P_c }{P_s \mathbb{E} \left( |h_{i+k}|^2 \right) + N_0 }  }  f_{i,i+k} y_{i} + n_{i,i+k},
\end{align}
where $f_{i,i+k}$ is the complex link gain, $n_{i,i+k}$ is the CSCG noise with distribution $\mathcal{CN}(0, N_0)$, and $P_c$ is the transmit power at the conferencing links. Here, the constant coefficient $\sqrt{ \frac{ P_c }{P_s \mathbb{E} \left( |h_{i+k}|^2 \right) + N_0 }  } $ is used to satisfy the average transmit power constraint of the conferencing link. Due to the out-of-band and possible wired conferencing link assumptions, we assume that $f_{i,i+k}$ is a fixed positive constant and uniformly and positively bounded (similarly as $\mathbb{E} \left( |h_i|^2 \right)$). Since the inputs of conferencing links may not be Gaussian, we adopt the transmit SNR $\frac{P_c}{N_0}$ as the quality metric of the conferencing links for convenience, instead of the rate constraints as in \cite{chuan2}.

In the second hop, $x_i$ with unit average power is
transmitted from the $i$-th relay to the destination, and the received signal $y$ at the destination is given as
\begin{eqnarray} \label{receive_signal_destination}
y=\sum_{i=1}^{N} {\sqrt{P_r} g_{i} x_i} + n,
\end{eqnarray}
where $g_{i}$ is the complex channel gain of the $i$-th relay-to-destination link, $P_r$ is the transmit power at each relay, and $n$ is the CSCG noise with distribution $\mathcal{CN} (0,N_0)$. We also assume that $g_i$'s are independent, of zero mean, and with uniformly and positively bounded $\mathbb{E} \left( |g_i|^2 \right)$ and $\mathbb{E} \left( |g_i|^4 \right) $.

Moreover, we assume that only the local CSIs are available at each relay: For the $i$-th relay, it knows the CSIs of the links directly connected with it, i.e., $h_{i-k}$, $f_{i-k,i},~k=0,1,\cdots,M$, and $g_i$.

\section{Capacity Upper Bound and Achievable Rates}

In this section, we exam the capacity upper bound and the achievable rates of the considered networks with the DF and AF relaying schemes, respectively. Moreover, we prove some capacity-achieving results under special conditions.

\subsection{Preliminary Results and Capacity Upper Bound}

In this subsection, we first present some preliminary results and the capacity upper bound.

\begin{Lemma} \label{sim_lemma1}
Let $\left\{ X_i \geq 0,~i=1,\cdots,N \right\}$ be independent random variables, whose means and variances are uniformly and positively bounded, respectively. Then, we have
\begin{align}
& \log \left( 1+  \sum_{i=1}^N X_i \right) - \log \left( 1+  \sum_{i=1}^N \mathbb{E} \left( X_i \right) \right) \xrightarrow{w.p.1} 0, \label{lemma1} \\
& \log \left(   \sum_{i=1}^N X_i \right) - \log \left(   \sum_{i=1}^N \mathbb{E} \left( X_i \right) \right) \xrightarrow{w.p.1} 0. \label{lemma2}
\end{align}
\end{Lemma}
\begin{proof}
By the Corollary 2.3 in \cite{coso2}, we have (\ref{lemma1}); and we could obtain (\ref{lemma2}) similarly.
\end{proof}

Using this lemma and the classic BC cut-set bound \cite{cover}, we obtain the following capacity upper bound.

\begin{Theorem} \label{upper_bound}
(BC cut-set bound) The capacity upper bound for the two-hop large Gaussian relay network is given as
\begin{align}
C_{\text{upper}} &  \leq  \frac{1}{2} \log \left( 1 + \frac{P_s}{N_0} \sum_{i=1}^N |h_i|^2 \right) \label{UPPER_pre} \\
  &  \xrightarrow{w.p.1} \frac{1}{2} \log \left( 1 + \frac{P_s}{N_0} \sum_{i=1}^N \mathbb{E} \left( |h_i|^2 \right) \right) \label{UPPER_sim} \\
 &  \sim \mathcal{O} \left( \log(N) \right) \label{upper_scaling}
\end{align}
\end{Theorem}
\begin{proof}
(\ref{UPPER_pre}) is by the similar result in \cite{blocskei}, and (\ref{UPPER_sim}) is by (\ref{lemma1}). Let $\mu = \frac{1}{N} \sum_{i=1}^N \mathbb{E} \left( |h_i|^2 \right)$, which is positively bounded, and we obtain (\ref{upper_scaling}).
\end{proof}

\subsection{The DF Achievable Rate}

In \cite{coso}, the authors showed that the DF rate scales at most on the order of $\mathcal{O} \left( \log(\log(N))\right)$ without conferencing among the relays, where the source chooses an optimal a subset of relays to decode the source message and let the rest keep silent in the second hop transmission. In this subsection, we adopt a different scheme to require all the relays to decode the source message and transmit in the second hop. Obviously, compared to the previous scheme \cite{coso}, our scheme is not optimal in term of relay subset selection, while it is enough to show the improvement of the achievable rate scaling behavior introduced by relay conferencing. Note that both the schemes in \cite{coso} and our proposed DF scheme require full channel CSI at the source node. Our main results for the DF relaying scheme is given as the following theorem.

\begin{Theorem}
Using the $p$-portion conferencing strategy, the DF rate scales on the order of $\mathcal{O} \left(  \log (N) \right)$.
\end{Theorem}
\begin{proof}
Based on the principle of maximum ratio combining (MRC), the received SNR in the relay is the sum of the SNRs in (\ref{t1}) and (\ref{tc}). Thus, for the first hop, the maximum rate supported at the $i$-th relay is given as
\begin{align}
R_i & = \frac{1}{2} \log \left( 1 + \frac{|h_i|^2 P_s}{N_0} + \frac{P_s}{N_0} \sum_{k=1}^M \frac{ \frac{ P_c}{P_s \mathbb{E} \left( |h_{i-k}|^2 \right) + N_0 } |f_{i-k,i}|^2 |h_{i-k}|^2  } {  \frac{ P_c}{P_s \mathbb{E} \left( |h_{i-k}|^2 \right) + N_0 } |f_{i-k,i}|^2 + 1  } \right) \\
& \xrightarrow{w.p.1} \frac{1}{2} \log \left( 1+  \frac{P_s}{N_0} \left( \mathbb{E} \left( |h_{i}|^2 \right) +  \sum_{k=1}^M  \frac{  P_c |f_{i-k,i}|^2 \mathbb{E} \left( |h_{i-k}|^2 \right) } { P_c |f_{i-k,i}|^2 + P_s \mathbb{E} \left( |h_{i-k}|^2 \right) + N_0  }   \right)  \right) \label{DF_rate_asym} \\
& = \frac{1}{2} \log \left( 1 + (M+1) \frac{P_s}{N_0} \mu_{\text{DF}} \right),
\end{align}
where (\ref{DF_rate_asym}) is by the Lemma \ref{sim_lemma1}, and $\mu_{\text{DF}} = \frac{1}{M+1}\left[ \mathbb{E} \left( |h_{i}|^2 \right) +  \sum_{k=1}^M  \frac{  P_c |f_{i-k,i}|^2 \mathbb{E} \left( |h_{i-k}|^2 \right) } { P_c |f_{i-k,i}|^2 + P_s \mathbb{E} \left( |h_{i-k}|^2 \right) + N_0  } \right]$, which is positively bounded. Thus, we have $R_i \sim \mathcal{O} \left( \log (N) \right)$.

In the second hop, we assume that all relays transmit simultaneously, and the transmit signal in the $i$-th relay is $x_i = \sqrt{ \frac{P_r}{\mathbb{E} \left( |g_i|^2 \right)}  } g_i^* x $. Thus, the received signal at the destination is given as
\begin{align}
y = \underbrace{ \sum_{i=1}^N \sqrt{\frac{P_r}{\mathbb{E} \left( |g_i|^2 \right)}} |g_i|^2}_{Q_0} x + n,
\end{align}
and the maximum rate supported in the second hop is given as
\begin{align}
R_{\text{MAC}} & = \frac{1}{2} \log \left( 1 +  \frac{Q_0^2}{N_0} \right) \\
& \sim \log \left(  \frac{Q_0}{\sqrt{N_0}} \right) \label{DF_MAC_sym2} \\
& \xrightarrow{w.p.1} \log \left(  \frac{\mathbb{E} \left( Q_0 \right) }{\sqrt{N_0}} \right) \label{DF_MAC_sym}\\
& = \frac{1}{2} \log \left( \frac{ P_r}{N_0} N^2 \mu^2 \right),
\end{align}
where (\ref{DF_MAC_sym2}) is valid as $N \rightarrow \infty$, (\ref{DF_MAC_sym}) is by (\ref{lemma2}), $ \mathbb{E} \left( Q_0 \right) = \sqrt{P_r} \sum_{i=1}^N \frac{\mathbb{E} \left( |g_i|^2 \right)}{\sqrt{\mathbb{E} \left( |g_i|^2 \right) }} = \sqrt{P_r} \sum_{i=1}^N \sqrt{ \mathbb{E} \left( |g_i|^2 \right) }$, and $\mu = \frac{1}{N} \sum_{i=1}^N \sqrt{ \mathbb{E} \left( |g_i|^2 \right) }$.

Therefore, the DF achievable rate is given as
\begin{align}
R_{\text{DF}} & = \min \left\{ \min_{i} \{R_i\} , R_{\text{MAC}}  \right\}.
\end{align}
Since $R_i$ and $ R_{\text{MAC}}$ scales as $\mathcal{O} \left( \log \left( N \right) \right)$ and $\mathcal{O} \left( \log \left( N^2 \right) \right)$, respectively, $R_{\text{DF}}$ scales with the order of $\mathcal{O} \left( \log \left( N \right) \right)$.
\end{proof}

\begin{Remark}
For the complete conferencing scheme, i.e., $M=N-1$, the DF scheme is not capacity-achieving, since the SNR penalty term $\frac{  P_c |f_{i-k,i}|^2 } { P_c |f_{i-k,i}|^2 + P_s \mathbb{E} \left( |h_{i-k}|^2 \right) + N_0  } $ is uniformly and positively bounded and strictly less than 1. For the case $0<p<1$, obviously, the DF scheme is also not capacity-achieving, and suffers another $(1-p)$-portion power gain loss.
\end{Remark}

\subsection{AF Achievable Rate}

In this subsection, we discuss the AF relaying scheme. Since we assume no global CSIs at the relays, the network-wide optimal combining at the relays as proposed in \cite{chuan2} cannot be deployed. Thus, with only local CSIs, MRC across conferencing signals is another good choice, which maximizes the received SNR at the relays. Unfortunately, MRC makes the rate expression too complicated to obtain any clean results. Instead, here we combine the received signals $y_i$ and $y_{i-k,i}$'s at the $i$-th relay as
\begin{align}
t_i =  h_i^* y_i + \sum_{k=1}^M  \sqrt{ \frac{ P_s \mathbb{E} \left( |h_{i-k}|^2 \right) +N_0}{ P_c } }  \frac{1}{f_{i-k,i}} h_{i-k}^* y_{i-k,i}.
\end{align}
Then, the transmit signal at the $i$-th relay is given as
\begin{align}
x_i & = a_i \sqrt{P_r} g_i^* t_i \\
 & = a_i \sqrt{P_r} g_i^* \left(  \sum_{k=0}^M  \sqrt{P_s} |h_{i-k}|^2 x   +  \sum_{k=0}^M h_{i-k}^* n_{i-k} + \sum_{k=1}^M \sqrt{ \frac{ P_s \mathbb{E} \left( |h_{i-k}|^2 \right) +N_0}{ P_c } }  \frac{h_{i-k}^* n_{i-k,i}}{f_{i-k,i}}  \right), \label{AF_relay_signal}
\end{align}
where $a_i$ is the power control factor to satisfy $\mathbb{E} (x_i) \leq  P_r$, and it is chosen as
\begin{align} \label{AF_power_constraint}
a_i^2 = \mathbb{E}^{-1} \left(|g_i|^{2} \right) \left[  P_s \mathbb{E} \left( \sum_{k=0}^M |h_{i-k}|^2 \right)^2  + \sum_{k=0}^M \mathbb{E} \left( |h_{i-k}|^2 \right) + \sum_{k=1}^M \frac{P_s \mathbb{E}\left( |h_{i-k}|^2\right) + N_0}{P_c |f_{i-k,i}|^2 } \mathbb{E} \left( |h_{i-k}|^2 \right) \right]^{-1}.
\end{align}

\begin{Remark}
This combining scheme is not valid for the case without relay conferencing, i.e., the conferencing link SNR $\frac{P_c}{N_0} = 0$. Moreover, if $|f_{i,i+k}|$ or $\frac{P_c}{N_0}$ is close to zero, it will boost the conferencing link noise $n_{i,i+k}$, which may make the performance even worse than the case without conferencing. However, our analysis will show that for uniformly and positively bounded $|f_{i,i+k}|$'s and arbitrary $\frac{P_c}{N_0}$, the AF scheme performs well as $N \rightarrow \infty$.
\end{Remark}

Based on (\ref{receive_signal_destination}) and (\ref{AF_relay_signal}), the received signal at the destination is given as
\begin{align}
y & = \sum_{i=1}^N g_i x_i + n \\
& =  \sqrt{P_r P_s} \underbrace{ \sum_{i=1}^N a_i |g_i|^2 \left( \sum_{k=0}^M  |h_{i-k}|^2 \right) }_{Q_1} x + \sqrt{P_r} \sum_{i=1}^N \left( \sum_{k=0}^M a_{i+k} |g_{i+k}|^2   \right) h_{i}^* n_i \nonumber \\
&~~~~~ + \sqrt{P_r} \sum_{i=1}^N \sum_{k=1}^M \sqrt{ \frac{ P_s \mathbb{E} \left( |h_{i-k}|^2 \right) +N_0}{ P_c } }  \frac{1}{f_{i-k,i}}  a_i  |g_i|^2 h_{i+k}^* n_{i,i+k} + n.
\end{align}
Then, the AF achievable rate is given as
\begin{align}
R_{\text{AF}} = \frac{1}{2} \log \left( 1 + \frac{ P_s P_r Q_1^2}{ \left( P_r Q_2 + P_r Q_3 +1 \right) N_0  } \right),
\end{align}
where
\begin{align}
Q_2 & = \sum_{i=1}^N \left(  \sum_{k=0}^M a_{i+k} |g_{i+k}|^2 \right)^2 | h_{i}|^2, \\
Q_3 & = \sum_{i=1}^N \sum_{k=1}^M |a_i|^2   \frac{ P_s \mathbb{E} \left( |h_{i-k}|^2 \right) +N_0}{ P_c |f_{i-k,i}|^2 } |g_i|^4  |h_{i-k}|^2.
\end{align}

Now we have
\begin{align}
&~~~~ \log \left( 1 + \frac{ P_s P_r Q_1^2}{ \left( P_r Q_2 + P_r Q_3 +1 \right) N_0  } \right) \\ &  \sim  \log \left(  \frac{ P_s P_r Q_1^2}{ \left( P_r Q_2 + P_r Q_3 +1 \right) N_0  } \right) \label{AF_p_appro} \\
& = 2 \log \left( \sqrt{\frac{P_s P_r}{ N_0} } Q_1 \right) - \log \left( P_r Q_2 + P_r Q_3 +1 \right)  \\
 & \xrightarrow{w.p.1} 2 \log \left( \sqrt{\frac{P_s P_r}{ N_0} } \mathbb{E} \left( Q_1 \right) \right) - \log \left( P_r \mathbb{E} \left( Q_2 \right) + P_r \mathbb{E} \left( Q_3 \right) +1 \right), \label{AF_p_symp} \\
& \sim  \log \left(  1 + \frac{ P_s P_r \mathbb{E}^2 \left( Q_1 \right) }{ \left( P_r  \mathbb{E} \left( Q_2 \right) + P_r \mathbb{E} \left( Q_3 \right) +1 \right) N_0  } \right), \label{AF_p_appro2}
\end{align}
where (\ref{AF_p_symp}) is by the Lemma \ref{sim_lemma1}. Notice that (\ref{AF_p_appro}) and (\ref{AF_p_appro2}) are valid since we only add or ignore a constant term, which can be neglected in the case of $N \rightarrow + \infty$.

As $N \rightarrow + \infty$, we have
\begin{align}
\mathbb{E} \left( Q_1 \right) & = \sum_{i=1}^N a_i \mathbb{E} \left( |g_i|^2 \right) \left(  \sum_{k=0}^M  \mathbb{E} \left( |h_{i-k}|^2 \right) \right) =  N (M+1) \mu_1, \\
\mathbb{E} \left( Q_2 \right) & = \sum_{i=1}^N \mathbb{E} \left( \left( \sum_{k=0}^M a_{i+k}  |g_{i+k}|^2 \right)^2 \right) \mathbb{E} \left(| h_{i}|^2 \right) = N(M+1)^2 \mu_2, \\
\mathbb{E} \left( Q_3 \right) & = \sum_{i=1}^N \sum_{k=1}^M |a_i|^2   \frac{ P_s \mathbb{E} \left( |h_{i-k}|^2 \right) +N_0}{ P_c |f_{i-k,i}|^2 } \mathbb{E} \left( |g_i|^4 \right) \mathbb{E} \left(  |h_{i+k}|^2 \right) = N M \mu_3,
\end{align}
where
\begin{align}
\mu_1 & = \frac{1}{N} \sum_{i=1}^N a_i \mathbb{E} \left( |g_i|^2 \right) \left(  \frac{1}{M+1} \sum_{k=0}^M  \mathbb{E} \left( |h_{i-k}|^2 \right) \right), \\
\mu_2 & = \frac{1}{N} \sum_{i=1}^N \mathbb{E} \left( \left( \frac{1}{M+1} \sum_{k=0}^M a_{i+k}  |g_{i+k}|^2 \right)^2 \right) \mathbb{E} \left(| h_{i}|^2 \right), \\
\mu_3 &  = \frac{1}{N} \sum_{i=1}^N \frac{1}{M} \sum_{k=1}^M |a_i|^2   \frac{ P_s \mathbb{E} \left( |h_{i-k}|^2 \right) +N_0}{ P_c |f_{i-k,i}|^2 } \mathbb{E} \left( |g_i|^4 \right) \mathbb{E} \left(  |h_{i+k}|^2 \right).
\end{align}
Since we assume that $\mathbb{E} \left( |h_i|^2 \right)$, $\mathbb{E} \left( |g_i|^2 \right)$, and $\mathbb{E} \left( |g_i|^4 \right)$ are uniformly and positively bounded, $|a_i|$, $\mu_1$, $\mu_2$, and $\mu_3$ are also bounded and positive. For the $p$-portion conferencing scheme, since $\mathbb{E} \left( Q_3 \right)$ scales on a smaller order than $\mathbb{E} \left( Q_2 \right)$ as $N$ goes to infinity, we obtain the AF rate as
\begin{align}
R_{\text{AF}} \xrightarrow{w.p.1}  \frac{1}{2} \log \left( 1 + N \frac{\mu_1^2 }{\mu_2} \frac{P_s}{N_0} \right).
\end{align}

\begin{Remark}
The term $Q_3$ is the contribution of the conferencing link noises. Since $\frac{\mathbb{E} \left(Q_3\right)}{\mathbb{E} \left(Q_2 \right)} \rightarrow 0$, we conclude that for the $p$-portion conferencing scheme, the conferencing link noises are asymptotically negligible as $N \rightarrow +\infty$. This suggests that for large relay networks with AF, we do not need high quality conferencing links, i.e., even with small $\frac{P_c}{N_0}$, and the performance of the AF scheme is reasonably good for large $N$.
\end{Remark}

It is difficult to verify whether the AF scheme is capacity-achieving or not for the case with $0<p<1$ and generally distributed $h_i$'s and $g_i$'s. In the following, we prove two special capacity-achieving cases, which may be applied to many widely-used scenarios.

\begin{Theorem}
If $h_i$'s and $g_i$'s are i.i.d., respectively, the AF scheme asymptotically achieves the capacity upper bound (\ref{UPPER_sim}) as $N$ goes to infinity for arbitrary $0 < p <1$ and $\frac{P_c}{N_0}>0$.
\end{Theorem}
\begin{proof}
Since $h_i$'s and $g_i$'s are i.i.d., $\mathbb{E} \left( |h_i|^2 \right)$, $\mathbb{E} \left( |g_i|^2 \right)$, and $\mathbb{E} \left( |g_i|^4 \right)$ are identical over different $i$'s, respectively. Let us exam the term $\frac{\mu_1^2}{\mu_2}$, and we have
\begin{align}
\frac{\mu_1^2}{\mu_2} & = \frac{1}{N} \sum_{i=1}^N \mathbb{E} \left( |h_i|^2 \right) \frac{ \sum_{j=1}^N \mathbb{E} \left( |h_j|^2 \right) \left( \sum_{k=0}^M a_{i+k} \mathbb{E} \left( |g_{i+k}|^2 \right) \right) \left( \sum_{t=0}^M a_{j+t} \mathbb{E} \left( |g_{j+t}|^2 \right) \right) }{ \sum_{j=1}^N \mathbb{E} \left(| h_{j}|^2 \right) \mathbb{E} \left( \left(  \sum_{s=0}^M a_{j+s}  |g_{j+s}|^2 \right)^2 \right)  }  \\
& = \frac{\mathbb{E} \left( |h_i|^2 \right) }{N} \sum_{i=1}^N \underbrace{ \frac{ N \mathbb{E}^2 \left( |g_i|^2 \right) \left[ \left( \sum_{k=0}^M a_{i+k} \right)  \left(  \sum_{t=0}^M a_{t} \right) \right] } {   \mathbb{E}^2 \left( |g_i|^2 \right) \sum_{j=1}^N \sum_{s_1 \neq s_2} a_{j+s_1}a_{j+s_2}  +   \mathbb{E} \left( |g_i|^4 \right) M \sum_{j=1}^N  a_{j}^2  } }_{C_i}.
\end{align}
From (\ref{AF_power_constraint}), we have $a_i^2 \approx \frac{1}{ \mathbb{E} \left( |h_i|^2 \right) \mathbb{E} \left( |g_i|^2 \right) P_s M^2  }$ for large $M$, and we have
\begin{align}
C_i  \approx \frac{\mathbb{E}^2 \left( |g_i|^2 \right) N (M+1)^2 }{  \mathbb{E}^2 \left( |g_i|^2 \right) N M(M+1)  +   \mathbb{E} \left( |g_i|^4 \right) M N   }  \rightarrow 1.
\end{align}
Hence, we have $\frac{\mu_1^2}{\mu_2} \rightarrow \mathbb{E} \left( |h_i|^2 \right)$. Therefore, the theorem is proved.
\end{proof}

\begin{Theorem}
For independent but not necessarily identically distributed $h_i$'s or $g_i$'s, the full conferencing scheme, i.e., $N=M+1$, asymptotically achieves the capacity upper bound as $N$ goes to infinity for arbitrary $\frac{P_c}{N_0} >0 $.
\end{Theorem}
\begin{proof}
For the complete conferencing scheme, we obtain
\begin{align}
Q_1  = \left( \sum_{i=1}^N a_i |g_i|^2 \right) \sum_{k=1}^N  |h_{k}|^2 ,~Q_2  =  \left( \sum_{i=1}^N a_{i} |g_{i}|^2 \right)^2 \sum_{k=1}^N  | h_{k}|^2
\end{align}
By a similar argument as in the previous theorem, we can show $\frac{P_r Q_3 + 1}{P_r Q_2} \xrightarrow{w.p.1} 0$ as $N$ goes to infinity such that we obtain
\begin{align}
R_{\text{AF}} & = \frac{1}{2}  \log \left( 1 + \frac{ P_s \sum_{k=1}^N  |h_{k}|^2 }{ \left( 1 + \frac{P_r Q_3 + 1}{P_r Q_2} \right) N_0  } \right) \\
& \xrightarrow{w.p.1} \frac{1}{2}  \log \left( 1 + \frac{ P_s  }{ N_0  } \sum_{k=1}^N  |h_{k}|^2 \right).
\end{align}
Therefore, the capacity upper bound is asymptotically achieved.
\end{proof}

\section{Numerical Results}

In this section, we present some simulation and numerical results to compare the performance among the proposed coding schemes. For simplicity, we assume that $h_i$'s and $g_i$'s are i.i.d. complex Gaussian random variable of $\mathcal{CN} (0,1)$, $|f_{i,i+k}|=1$, $P_s = 1$, $P_r = 1$, and $N_0=1$. The rates in all the simulations, are averaged over 1000 fading realizations.

In Fig. \ref{fig3}, we show the capacity upper bound and the achievable rates for different $p$ values, as the number of relays increases. For the AF relaying scheme, the gap between the upper bound and the achievable rate is very small for $p=0.2$ and large $N$ values. For the DF relaying scheme, when $N$ is large, we observe that the DF rate and the capacity upper bound have the same scaling behavior.

In Fig. \ref{fig4}, we plot the achievable rates as functions of $p$. For the AF relaying scheme, the $p$ value does not need to be large to achieve most of the gains, i.e., around $p=0.3$; on the other hand, conferencing may not strictly improve the AF rate: When $p$ is close to zero, the achievable rate is lower than the case without relay conferencing, which is due to the sub-optimality of the combining scheme at the relays. For the DF relaying scheme, relay conferencing always helps, and there is a significant rate improvement as $p$ increases.

In Fig. \ref{fig5}, we plot the achievable rates as functions of the conferencing link SNR. It is observed that with medium-quality conferencing links (the SNRs of the conferencing links are around 5 dB), we achieve most of the gains introduced by relay conferencing for both the AF and DF relaying schemes.

\section{Conclusion}
In this correspondence, we investigated the achievable rate scaling laws of the DF and AF relaying schemes in a large Gaussian relay networks with conferencing links. We showed that for the DF relaying scheme, the rate scales as $\mathcal{O} \left( \log (N) \right)$, compared to $\mathcal{O} \left( \log (\log(N)) \right)$ for the case without conferencing; for the AF relaying scheme, we proved that if the channel fading coefficients $h_i$'s and $g_i$'s are i.i.d., respectively, or $N=M+1$, it asymptotically achieves the capacity upper bound as $N$ goes to infinity.

% \fi
\IEEEpeerreviewmaketitle



\begin{figure}[h]
\centering
\includegraphics[width=.7\linewidth]{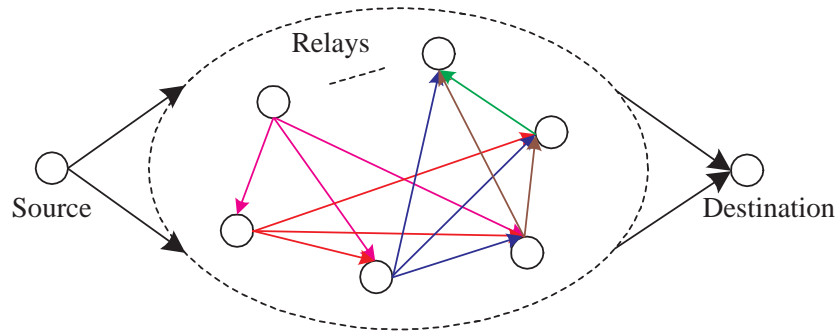}
\caption{The large relay networks with conferencing links.} \label{fig1}
\end{figure}

\begin{figure}[h]
\centering
\includegraphics[width=.7\linewidth]{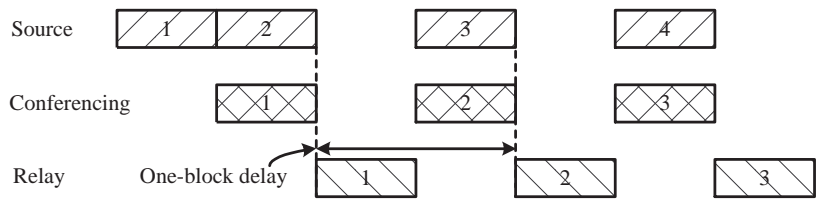}
\caption{Transmission scheduling scheme for the large relay networks with conferencing links.} \label{fig2}
\end{figure}

\begin{figure}[h]
\centering
\includegraphics[width=.7\linewidth]{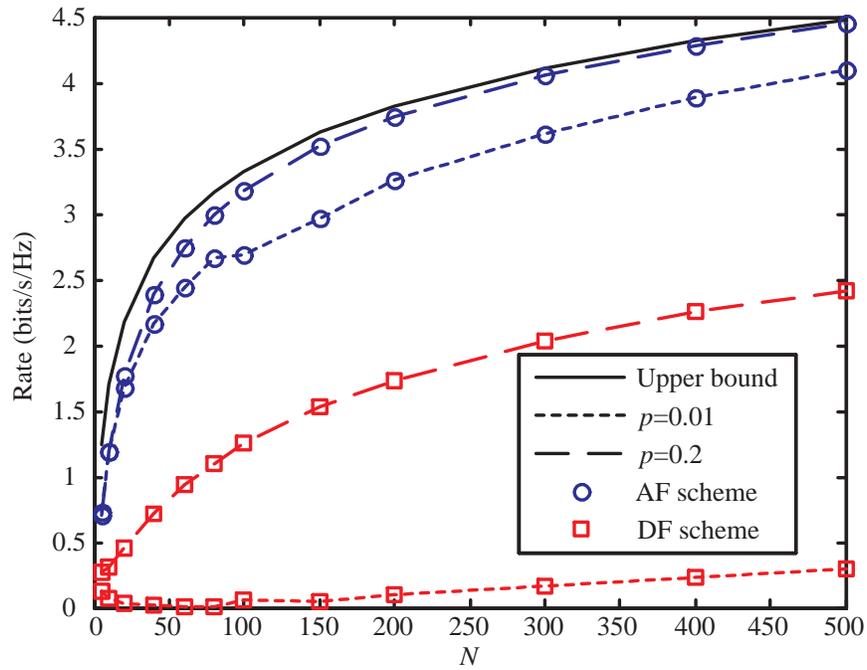}
\caption{Achievable rates vs. the number of relays, $P_s=1$, $P_r=1$, $P_c=1$, and $|f_{i,k}|=1$.} \label{fig3}
\end{figure}

\begin{figure}[h]
\centering
\includegraphics[width=.7\linewidth]{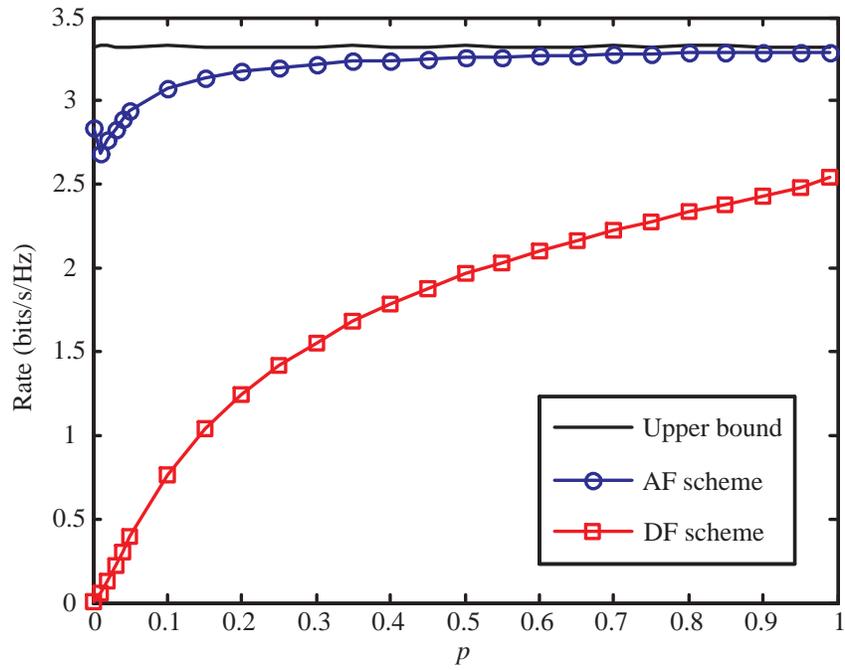}
\caption{Achievable rates vs. the conferencing ratio, $P_s=1$, $P_r=1$, $P_c=1$, $|f_{i,k}|=1$, and $N=100$.} \label{fig4}
\end{figure}

\begin{figure}[h]
\centering
\includegraphics[width=.7\linewidth]{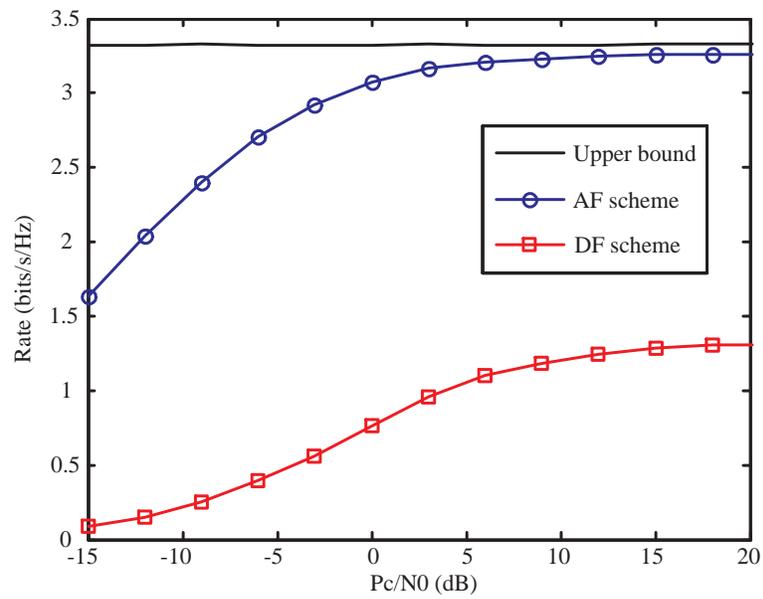}
\caption{Achievable rates vs. the conferencing link SNR, $P_s=1$, $P_r=1$, $|f_{i,k}|=1$, $N=100$, and $p=0.1$.} \label{fig5}
\end{figure}



%\bibliographystyle{IEEEtran}

\end{document}